\documentclass[12pt]{article}
\usepackage{amsfonts}
\usepackage{amssymb,amsthm,latexsym}
\usepackage[numbers,compress]{natbib}
\usepackage{amsmath}
\usepackage{indentfirst}
\allowdisplaybreaks[4]
\newtheorem{theorem}{Theorem}[section]

\newtheorem{cor}[theorem]{Corollary}
\newtheorem{lemma}[theorem]{Lemma}

\newtheorem{define}[theorem]{Definition}
\newtheorem{example}[theorem]{Example}

\numberwithin{equation}{section}
\begin{document}

\title{Self-Dual Codes over  $\mathbb{Z}_2\times (\mathbb{Z}_2+u\mathbb{Z}_2)$  }
%\author{Long Yu, Qiong Huang, Hongwei Liu,{\thanks{Corresponding author. \newline \indent ~~Email addresses:~hwliu@mail.ccnu.edu.cn~(Hongwei Liu),~longyuhbpu@163.com~(Long Yu), 1247698055@qq.com (Qiong Huang).} }}
\author{Long Yu$^{1,2}$, Qiong Huang$^{2}$, Hongwei Liu$^{2}${\thanks{Corresponding author.
\newline \indent ~~Email addresses:~hwliu@mail.ccnu.edu.cn~(Hongwei Liu),~longyuhbpu@163.com~(Long Yu), 1247698055@qq.com~(Qiong Huang), lxs6682@163.com~(Xiusheng Liu).} }, Xiusheng Liu$^{1}$}
\date{$1$  School of Mathematics and Physics, Hubei Polytechnic University, Huangshi,  435003, China\\
$2$ School of Mathematics and Statistics, Central China Normal University, Wuhan,  430079, China}
\maketitle
%\rule{\textwidth}{0.5pt}
\begin{abstract}
In this paper, we study self-dual codes over  $\mathbb{Z}_2 \times (\mathbb{Z}_2+u\mathbb{Z}_2) $, where $u^2=0$. Three types of self-dual codes are defined. For
each type, the possible values $\alpha,\beta$   such that there exists a code $\mathcal{C}\subseteq \mathbb{Z}_{2}^\alpha\times (\mathbb{Z}_2+u\mathbb{Z}_2)^\beta$ are established. We  also present several approaches to  construct  self-dual codes over  $\mathbb{Z}_2 \times (\mathbb{Z}_2+u\mathbb{Z}_2) $.  Moreover, the structure of  two-weight self-dual codes is completely obtained for $\alpha \cdot\beta\neq 0$.
\end{abstract}

%%%%%%%%%%%%%%%%%%%%%%%%%%%%%%%%%%%%%%%%%%%%%%%%%%%%%%%%%%%%%%%%%%%%%%%%%%

{\bf Key Words}\ \ Linear code, Self-dual codes, Two-weight self-dual codes\\

%{\bf Mathematics Subject Classification}
\section{Introduction}
A binary code $C$  of length $n$ over the finite field $\mathbb{Z}_2$ is a subset of $\mathbb{Z}_{2}^{n}$. If $C$ is a vector subspace of $\mathbb{Z}_{2}^{n}$, then we call this
code is  linear. A quaternary code $\mathcal{C}$ is a subset of $\mathbb{Z}_{4}^{n}$ and it is said to be linear if it is a submodule. In \cite{Delsarte1973}, Delsarte proposed the definition of additive codes, which are subgroups of the underlying abelian group in a translation association scheme. In particular, a binary Hamming scheme which is the only structures for the abelian group are those of the form $\mathbb{Z}_{2}^{\alpha}\times \mathbb{Z}_{4}^{\beta}$ when the underlying abelian group is of order $2^{n}$, where $n=\alpha+2\beta$(\cite{Delsarte1998}). Hence, the only additive codes in a binary Hamming scheme are the subgroups $\mathcal{C}$ of $\mathbb{Z}_{2}^{\alpha}\times \mathbb{Z}_{4}^{\beta}$. In order to distinguish them from additive codes over finite fields (see \cite{Bachoc2000,Bierbrauer2005,Blokhuis2004,J-L.Kim2003}), we call them $\mathbb{Z}_{2}\mathbb{Z}_{4}$-additive codes. Foundational results on $\mathbb{Z}_{2}\mathbb{Z}_{4}$-additive codes, including the generator matrix, the existence and the construction of self-dual codes, can be found in \cite{Borges2009}.

Now, we present another important ring $R$ which contains four elements, where $R=\mathbb{Z}_{2}+u\mathbb{Z}_{2}=\{0,1,u,1+u\}$ with $u^{2}=0$. It is well know that the ring $\mathbb{Z}_{2}$ is a subring of the ring $R$. Similar as that $\mathbb{Z}_{2}\mathbb{Z}_{4}-$additive codes, we  define  the following set:
$$\mathbb{Z}_{2}^\alpha \times R^\beta=\{(\mathbf{a},\mathbf{b})\mid~\mathbf{a}\in \mathbb{Z}_{2}^\alpha,~\mathbf{b}\in R^\beta\}.$$

We cannot directly define an algebraic structure endowing the set $\mathbb{Z}_{2}R$, which is not well defined with respect to the usual multiplication by $u\in R$. So this set is not $R$-module. In order to make it well defined and enrich with an algebraic structure we introduce a new multiplication as follows.

Define a map
$$\eta:R\longrightarrow \mathbb{Z}_{2},\eta (r+uq)\longmapsto r,$$
Clearly, the map $\eta$ is a ring homomorphism. Using this map, we can define a scalar multiplication as follows: for $\nu=(a_{1},a_{2},\dots,a_{\alpha}\mid b_{1},b_{2},\dots,b_{\beta})\in \mathbb{Z}_{2}^\alpha\times R^{\beta}$ and $d\in R$, we have
\begin{equation}\label{eq:1.1q}
d\nu=(\eta(d)a_{1},\cdots,\eta(d)a_{\alpha}\mid db_{1},\cdots,db_{\beta}).
\end{equation}

\begin{define} A linear code $\mathcal{C}$ is called a $\mathbb{Z}_2  R$ linear codes if it is a $R$-submodule of $\mathbb{Z}_2^\alpha \times R^\beta$ with respect to the scalar multiplication defined in \eqref{eq:1.1q}. Then the binary image of $\Phi(\mathcal{C})=C$ is called a  $\mathbb{Z}_2 R$-linear code of length $n=\alpha+2\beta$, where $\Phi$ is a map from $\mathbb{Z}_2^\alpha \times R^\beta$ to $\mathbb{Z}_2^n$ defined as \[\Phi(a,b)=(a_1,\cdots,a_{\alpha}\mid \phi(b_1),\cdots,\phi(b_{\beta})),\]
for all $a=(a_0,\cdots,a_{\alpha})\in \mathbb{Z}_2^\alpha$, and $b=(b_1,\cdots,b_{\beta})\in R^\beta$. Furthermore, $\phi: R$ to $\mathbb{Z}_2^2$ is defined by $\phi(0)=(0,0)$, $\phi(1)=(0,1)$, $\phi(u)=(1,1)$, $\phi(1+u)=(1,0)$.
\end{define}

With  the above preparation, Aydogdu et al. in \cite{Ismail Aydogdu2014} obtained the standard form matrix of $\mathbb{Z}_2R$ linear codes.
\begin{theorem}\label{th:standard form}{\rm \cite{Ismail Aydogdu2014}}
Let $\mathcal{C}$ be a $\mathbb{Z}_{2}R$ linear code of type $(\alpha,\beta;\gamma,\delta;\kappa)$.Then $\mathcal{C}$ is permutation equivalent to a $\mathbb{Z}_{2}R$ linear code with the standard form matrix
\begin{equation}\label{eq:standard matrix}
   G =\left(
  \begin{array}{cc|ccc}
    I_{\kappa} & A_{1} & uT & 0 & 0  \\
    0 & 0 & uD & uI_{\gamma-\kappa} & 0 \\
    0 & S & B_{1}+uB_{2} & A & I_{\delta}\\
  \end{array}
\right),
\end{equation}
where $A,A_{1},B_{1},B_{2},D,S$ and $T$ are matrices over $\mathbb{Z}_{2}$.
\end{theorem}

From Theorem~\ref{th:standard form}, it is easy  that  see $\mathbb{Z}_{2}R$ linear code is isomorphic to $\mathbb{Z}_{2}^{\gamma}\times \mathbb{Z}_{2}^{2\delta}$, and it has $|\mathcal{C}|=2^{\gamma+2\delta}$. Moreover, having the generator matrix as above, we say that $\mathcal{C}$ is of type $(\alpha,\beta;\gamma,\delta;\kappa)$. Also, $\kappa$ can be defined as follows:

Let $X$(respectively $Y$) be the set of $\mathbb{Z}_{2}$ (respectively $R$) coordinates positions, hence $|X|=\alpha$ (respectively $|Y|=\beta$). Unless otherwise stated, the $X$ corresponds to the first $\alpha$ coordinates and $Y$ corresponds to the last $\beta$ coordinates. Call $\mathcal{C}_{X}$ (respectively $\mathcal{C}_{Y}$) the punctured code of $\mathcal{C}$ by deleting the coordinates outside $X$ (respectively $Y$). Let $\mathcal{C}_{b}$ be the subcode of $\mathcal{C}$ which contains all codewords having the form of $(x|y_{1},y_{2},\dots,y_{\beta})$, where $x\in \mathbb{Z}_{2}^{\alpha},y_{i}\in \{0,u\},i=1,2,\dots,\beta$. Then $\kappa=dim(\mathcal{C}_{b})_{X}$. For the case $\alpha=0$, we will take $\kappa=0$.

\begin{define}An inner product for two vectors $\mathbf{v}=(v_1,\cdots,v_\alpha|v_{\alpha+1},\cdots,v_{\alpha+\beta}),  \mathbf{w}=(w_1,\cdots,w_\alpha|w_{\alpha+1},\cdots,w_{\alpha+\beta})\in \mathbb{Z}_{2}^{\alpha}\times (\mathbb{Z}_{2}+u\mathbb{Z}_{2})^{\beta}$ is defined as
\begin{equation}\label{eq:innerproduct}
\langle \mathbf{v},\mathbf{w}\rangle=u\left(\sum\limits_{i=1}^\alpha v_{i}w_{i}\right)+\sum\limits_{j=\alpha+1}^{\alpha+\beta} v_{j}w_{j}\in R=\mathbb{Z}_{2}+u\mathbb{Z}_{2}.
\end{equation}
\end{define}

Hence, we have the  definition of dual codes.
\begin{define}
Let $\mathcal{C}$ be a $\mathbb{Z}_{2}R$  code. We denote the  dual of $\mathcal{C}$ by $\mathcal{C}^\perp$, which is defined  as
$$\mathcal{C}^\perp=\{\mathbf{w}\in \mathbb{Z}_{2}^{\alpha}\times (\mathbb{Z}_{2}+u\mathbb{Z}_{2})^{\beta}\mid \langle \mathbf{v},\mathbf{w}\rangle=0~{\rm for~all}~ \mathbf{v}\in \mathcal{C}\}.$$
We say that $\mathcal{C}$ is self-orthogonal if and only if $\mathcal{C}\subseteq \mathcal{C}^\perp$ and $\mathcal{C}$ is  self-dual if and only if $\mathcal{C}=\mathcal{C}^\perp$.
\end{define}

%Let $C=\Phi(\mathcal{C})$ be the corresponding $\mathbb{Z}_{2}R-$linear code.
%We call $\mathcal{C}$ is a self $\mathbb{Z}_{2}R-$dual code if $C=C_{\perp}$ where $C_{\perp}=\mathcal{C}^\perp$.

Following above definitions, the standard form generator matrix of  $\mathcal{C}^\perp$ can be obtained.
\begin{theorem}\label{th:dualcodestandardform}{\rm \cite{Ismail Aydogdu2014}}
Let $\mathcal{C}$ be a $\mathbb{Z}_{2}R$ linear code of type $(\alpha,\beta;\gamma,\delta;\kappa)$ with standard form matrix defined in \eqref{eq:standard matrix}. Then the generator matrix of $\mathcal{C}^{\perp}$ is given as
$$
H=\left(
  \begin{array}{cc|ccc}
     A_{1}^{t}& I_{\alpha-\kappa} & 0 & 0 & uS^{t}  \\
    0 & 0 & 0 & uI_{\gamma-\kappa} & uA^{t} \\
    T^{t} & 0 & I_{\beta+\kappa-\gamma-\delta} & D^{t} &  (B_{1}+uB_{2})^{t}+D^{t}A^{t}\\
  \end{array}
\right),
$$
where $A,A_{1},B_{1},B_{2},D$ and $T$ are matrices over $\mathbb{Z}_{2}$.
\end{theorem}

From Theorem~\ref{th:dualcodestandardform}, we know that the  dual code $\mathcal{C}^{\perp}$ is  a $\mathbb{Z}_{2}R$ linear code of type $(\alpha,\beta;\bar{\gamma},\bar{\delta};\bar{\kappa})$, where
\[
\begin{cases}
\bar{\kappa}=\alpha-\kappa;\\
\bar{\gamma}=\alpha+\gamma-2\kappa;\\
\bar{\delta}=\beta-\gamma-\delta+\kappa\\
\end{cases}
\]

Let $(\textbf{v}| \textbf{w})=(v_{1},\dots,v_{\alpha}\mid w_{1},\dots,w_{\beta})\in \mathbb{Z}_{2}^{\alpha}\times R^{\beta}$. The Gray map defined on $R$ can be expressed as follows: $\phi(a+bu)=(b,a+b),a,b\in \mathbb{Z}_{2}$ and the Lee weight is $wt_{L}(a+bu)=wt_{H}(b,a+b)$. Hence, $wt_L(\textbf{v}|\textbf{w})=wt_{H}(\textbf{v} )+wt_{L}(\textbf{w})$.

Let $\mathcal{C}$ be a $\mathbb{Z}_2R $ linear code of type $(\alpha,\beta;\gamma,\delta;\kappa)$ and $n=\alpha+2\beta$. The weight enumerator of  code $\mathcal{C}$ is defined as
$$W(X,Y)=\sum\limits_{c\in \mathcal{C}}X^{n-wt_L(c)}Y^{wt_L(c)}.$$

Aydogdu et al. \cite{Ismail Aydogdu2014} gave the following result.
\begin{theorem}\label{th:MacIden}{\rm \cite{Ismail Aydogdu2014}}
Let $\mathcal{C}$ be a $\mathbb{Z}_{2}R$ linear code.The relation between the weight enumerators of $\mathcal{C}$ and $\mathcal{C}^\perp$  is given by the following identity:
$$W_{\mathcal{C}^{\perp}}(X,Y)=\frac{1}{\mid\mathcal{C}\mid}W_{\mathcal{C}}(X+Y,X-Y).$$
\end{theorem}

Based on  above  definitions, we can construct a self-dual code $\mathcal{C}_1$  over $\mathbb{Z}_{2}^4R^2$ of type $(4,2;2,1;2)$, which is generated by the following matrix
$$\left(\begin{array}{cccc|cc}
          1 & 0 & 1 & 0 & u & 0 \\
          0 & 1 & 0 & 1 & u & 0 \\
          0 & 0 & 1 & 1 & 1 & 1
        \end{array}
\right).$$
It is easy to get the weight enumerator of $\mathcal{C}_1$ is
\[W(X,Y)=X^{8}+14X^4Y^4+Y^{8}, \]
which implies that all codewords have doubly-even weight.

In \cite{Borgesselfdual}, Borges et al. studied self-dual codes over $\mathbb{Z}_{2}\mathbb{Z}_{4}$. Three types of self-dual codes are defined. For
each type, the possible values $\alpha,\beta$   such that there exists a code $\mathcal{C}\subseteq \mathbb{Z}_{2}^\alpha\times \mathbb{Z}_{4}^\beta$ are established.
Borges et al. defined a self-dual code is Type II if all the codewords have doubly-even weight. And they showed  if $\mathcal{C} $ is  type II  $\mathbb{Z}_{2}\mathbb{Z}_{4}$-additive code, then $\alpha\equiv 0~({\rm mod} ~8)$. Recall that linear code $\mathcal{C}_1$ over $\mathbb{Z}_{2}^4R^2$, it is easy to find this case dosenot exist in \cite{Borgesselfdual}. Motivated by this discovery, we furthermore study self-dual codes over $\mathbb{Z}_{2}R$. Similar as that in \cite{Borgesselfdual}, three types of self-dual codes are defined. We also give the existence condition for each type and present several approaches to construct self-dual codes. Finally, we study self-dual codes over $\mathbb{Z}_{2}R$ with two nonzero weights, and the structure of these codes is described.

This paper is organized as follows. Section $2$, we study the properties of self-dual codes over $\mathbb{Z}_2R$, and give the existence conditions for three types. In Section $3$, we give several approaches of constructing  self-dual codes over $\mathbb{Z}_2R$. Section $4$, we determine the structure of   two-weight self-dual codes over $\mathbb{Z}_{2}R$ for $\alpha\cdot\beta\neq0$.

\section{Self-dual codes over $\mathbb{Z}_2R$}

%In this section, we will give some necessary conditions on $\mathbb{Z}_2 \mathbb{Z}_2[u]$-additive self-dual codes. Furthermore, we give a sufficient and necessary condition of $\mathbb{Z}_2 \mathbb{Z}_2[u]$-additive self-dual codes in some special case.

\begin{lemma}\label{lem:selfdualcanshu}
Let $\mathcal{C}$ be a $\mathbb{Z}_2R$ linear self-dual code, then $\mathcal{C}$ is of type $(2\kappa,\beta;\beta+\kappa-2\delta,\delta;\kappa)$, $| \mathcal{C}|=2^{\kappa+\beta}$ and $\mathcal{C}_b=2^{\kappa+\beta-\delta}$.
\end{lemma}
\begin{proof} By Theorem~\ref{th:standard form} and Theorem~\ref{th:dualcodestandardform}, we finish the proof.\end{proof}

Hence, we have the following result.
\begin{cor}
Let $\mathcal{C}$ be  a $\mathbb{Z}_2R$ linear self-dual code of type $(\alpha,\beta;\gamma,\delta;\kappa)$ , then $\alpha$ and $n$ are both even.
\end{cor}

\begin{lemma}\label{wandNwareeven}
Let $\mathcal{C}$ be a $\mathbb{Z}_2R$ linear self-dual code and $(\mathbf{v}| \mathbf{w})\in \mathcal{C}$. Denote $N(\mathbf{w})$  as  the number of unit ($1$ or $1+u$) coordinates of vector $\mathbf{w}\in R^\beta$, then $wt_H(\mathbf{v})$ and $N(\mathbf{w})$ are both even. Moreover, we have $(\mathbf{1}^\alpha,\mathbf{0}^\beta)$, $(\mathbf{{0}}^\alpha,\mathbf{u}^\beta)$  and $(\mathbf{1}^\alpha,\mathbf{u}^\beta)$ are all in $\mathcal{C}$, where $\mathbf{a}^r$ is defined as the tuple $\overbrace{(a,a,\cdots,a)}^r$.
\end{lemma}
\begin{proof} Since $\mathcal{C}$ is a self-dual code, then for any codeword $(\mathbf{v}\mid \mathbf{w})\in \mathcal{C}$,   we have $\langle(\mathbf{v}| \mathbf{w}),(\mathbf{v}| \mathbf{w})\rangle=u\cdot wt_H(\mathbf{v})+N(\mathbf{w})=0\in R$. Note that  $wt_H(\mathbf{v})$ and $N(\mathbf{w})$ are all integers, so $wt_H(\mathbf{v})$ and $N(\mathbf{w})$ are both even. Since $wt_H(\mathbf{v})$ and $N(\mathbf{w})$ are both even, then $(\mathbf{1}^\alpha,\mathbf{0}^\beta)$, $(\mathbf{0}^\alpha,\mathbf{u}^\beta)$  are in $\mathcal{C}$. We are done.\end{proof}
\begin{lemma}\label{lem:CBXself}
Let $\mathcal{C}$ be a linear self-dual code, then the subcode $(\mathcal{C}_b)_X$ is a binary self-dual code.
\end{lemma}
\begin{proof} By Lemma~\ref{lem:selfdualcanshu}, we have $\mathcal{C}$ is of type $(2\kappa,\beta;\beta+\kappa-2\delta,\delta;\kappa)$. Note that for any two codewords $(\mathbf{x}|\mathbf{y})$, $(\mathbf{v}|\mathbf{w})\in \mathcal{C}_b$, one has $\langle\mathbf{y},\mathbf{w}\rangle=0$. This implies $(\mathcal{C}_b)_X\subseteq(\mathcal{C}_b)_X^\perp$. Since the dimension of $(\mathcal{C}_b)_X$ is $\kappa$ and the length of $(\mathcal{C}_b)_X$ is $\alpha=2\kappa$, we are done.\end{proof}
\begin{lemma}
Let $\mathcal{C}$ be a linear self-dual code of type $(2\kappa,\beta;\beta+\kappa-2\delta,\delta;\kappa)$. There exists an integer $r,$ $1\leq r\leq\kappa$, such that each codeword in $\mathcal{C}_Y$ appears $2^r$ times in $\mathcal{C}$ and $|\mathcal{C}_Y|\geq2^\beta$.
\end{lemma}
\begin{proof} Denote the subcode $\mathcal{C}_0=\{(\mathbf{v}|\mathbf{0})\in \mathcal{C}\}$. It is easy to see that $(\mathcal{C}_0)_X$ is a linear binary code with dimension $r=dim(\mathcal{C}_0)_X$. Thus, any vector in $\mathcal{C}_Y$ appears $2^r$ times in $\mathcal{C}$. Note that $(\mathcal{C}_0)_X\subseteq (\mathcal{C}_b)_X$, then $r\leq\kappa$. Since $|\mathcal{C}|=2^{\beta+\kappa}=|\mathcal{C}_Y||\mathcal{C}_0|$, then $|\mathcal{C}_Y|\geq2^\beta$.\end{proof}

For convenience, we define  notations in the following    for any two codewords $(\mathbf{v}| \mathbf{w})$, $(\mathbf{x}| \mathbf{y})\in \mathcal{C}$.
 \begin{table}[!h]
  \centering
 % \caption{some notations are fixed throughout this paper}
\vspace*{10pt}
\begin{tabular}{ l l }
  \hline
$N(\mathbf{w})$ & the number of units of $w$ \\
$N_u(\mathbf{w})$ & the number of  $u\in w$  \\
  $N_{1,1}(\mathbf{w},\mathbf{y})$ & $\# \{i \mid w_i=1 ~{\rm or}~ 1+u, y_i=1 ~{\rm or}~ 1+u, 1\leq i\leq \beta\}$\\
  $N_{1,u}(\mathbf{w},\mathbf{y})$ & $\# \{i \mid w_i=1 ~{\rm or}~ 1+u, y_i=u, 1\leq i\leq \beta\}$ \\
 $ N_{u,1}(\mathbf{w},\mathbf{y})$ &$ \# \{i \mid w_i=u, y_i=1 ~{\rm or}~ 1+u, 1\leq i\leq \beta\}$ \\
 $ N_s(\mathbf{w},\mathbf{y})$ &  $\#\{i \mid w_i=y_i=1 ~{\rm or}~ 1+u, 1\leq i\leq \beta\}$ \\
  $N_d(\mathbf{w},\mathbf{y})$ & $\#\{i \mid w_i=1, y_i=1+u  ~{\rm or}~ w_i=1+u, y_i=1, 1\leq i\leq \beta\}$ \\
  \hline
\end{tabular}
\end{table}

%\begin{center}
%\begin{tabular}{ c l }
%  \hline
%  % after \\: \hline or \cline{col1-col2} \cline{col3-col4} ...
%  $N_{1,1}(\mathbf{w},\mathbf{y})$ & $\# \{i \mid w_i=1 ~{\rm or}~ 1+u, y_i=1 ~{\rm or}~ 1+u, 1\leq i\leq \beta\}$\\
%  $N_{1,u}(\mathbf{w},\mathbf{y})$ & $\# \{i \mid w_i=1 ~{\rm or}~ 1+u, y_i=u, 1\leq i\leq \beta\}$ \\
% $ N_{u,1}(\mathbf{w},\mathbf{y})$ &$ \# \{i \mid w_i=u, y_i=1 ~{\rm or}~ 1+u, 1\leq i\leq \beta\}$ \\
% $ N_s(\mathbf{w},\mathbf{y})$ & \# $\{i \mid w_i=y_i=1 ~{\rm or}~ 1+u, 1\leq i\leq \beta\}$ \\
%  $N_d(\mathbf{w},\mathbf{y})$ & $\{i \mid w_i=1, y_i=1+u  ~{\rm or}~ w_i=1+u, y_i=1, 1\leq i\leq \beta\}$ \\
%  \hline
%\end{tabular}
%\end{center}
%$$N_{1,1}(\mathbf{w},\mathbf{y})=\# \{i \mid w_i=1 ~{\rm or}~ 1+u, y_i=1 ~{\rm or}~ 1+u, 1\leq i\leq \beta\},$$
%$$N_{1,u}(\mathbf{w},\mathbf{y})=\# \{i \mid w_i=1 ~{\rm or}~ 1+u, y_i=u, 1\leq i\leq \beta\},$$
%$$N_{u,1}(\mathbf{w},\mathbf{y})=\# \{i \mid w_i=u, y_i=1 ~{\rm or}~ 1+u, 1\leq i\leq \beta\},$$
%$$N_s(\mathbf{w},\mathbf{y})=\# \{i \mid w_i=y_i=1 ~{\rm or}~ 1+u, 1\leq i\leq \beta\}, $$
%$$N_d(\mathbf{w},\mathbf{y})=\# \{i \mid w_i=1, y_i=1+u  ~{\rm or}~ w_i=1+u, y_i=1, 1\leq i\leq \beta\}.$$
%It is easy to check that $N_{1,1}(\mathbf{w},\mathbf{y})=N_s(\mathbf{w},\mathbf{y})+N_d(\mathbf{w},\mathbf{y})$.
\begin{lemma}
Let $\mathcal{C}$ be a $\mathbb{Z}_2R$ linear self-dual code, and $(\mathbf{v}|\mathbf{w})$, $(\mathbf{x}| \mathbf{y})$ are two codewords in $\mathcal{C}$.  Then we have $N_s(\mathbf{w},\mathbf{y}) \equiv N_d(\mathbf{w},\mathbf{y}) ~({\rm mod}~ 2)$ and $N_{1,1}(\mathbf{w},\mathbf{y})$ is even.
\end{lemma}
\begin{proof} Let $(\mathbf{v}| \mathbf{w})=(v_1,\cdots,v_\alpha | w_1,\cdots,w_\beta)$, $(\mathbf{x}| \mathbf{y})=(x_1,\cdots,x_\alpha|y_1,\cdots,y_\beta)$.  In the following,  we consider the inner product of $\langle(\mathbf{v}| \mathbf{w}),(\mathbf{x}| \mathbf{y})\rangle$.  By \eqref{eq:innerproduct}, one has
\begin{eqnarray*}
&&\langle(\mathbf{v}| \mathbf{w}),(\mathbf{x}| \mathbf{y})\rangle\\
   &=&u\sum_{i=1}^\alpha v_ix_i+uN_{1,u}(\mathbf{w},\mathbf{y})+uN_{u,1}(\mathbf{w},\mathbf{y})+N_s(\mathbf{w},\mathbf{y})+N_d(\mathbf{w},\mathbf{y})+uN_d(\mathbf{w},\mathbf{y}) \\
   &=& u\left(\sum_{i=1}^\alpha v_ix_i+N_{1,u}(\mathbf{w},\mathbf{y})+N_{u,1}(\mathbf{w},\mathbf{y})+N_d(\mathbf{w},\mathbf{y})\right)+N_s(\mathbf{w},\mathbf{y})+N_d(\mathbf{w},\mathbf{y})\\
&=& 0\in R.
\end{eqnarray*}
%\[<(v\mid (1+u)w),(x\mid y)>=\sum_{i=1}^\alpha v_ix_i+uN_1+uN_2+N_s+N_d+uN_s=0\in\mathbb{Z}_2[u].\]
Then,  we have $N_s(\mathbf{w},\mathbf{y}) \equiv N_d(\mathbf{w},\mathbf{y})~({\rm mod}~ 2)$, which implies $N_{1,1}(\mathbf{w},\mathbf{y})$ is even.\end{proof}

\begin{define}
Let $\mathcal{C}$ be a $\mathbb{Z}_2R$ linear  code. If $\mathcal{C}=\mathcal{C}_X\times \mathcal{C}_Y$, then $\mathcal{C}$ is called separable.
\end{define}

By Theorem~\ref{th:standard form}, if $\mathcal{C}$ is a separable linear  code, we have that the generator matrix of $\mathcal{C}$ in standard form as follows
\[G_s=\left(\begin{array}{cc|ccc}
        I_\kappa & A & 0 &0 & 0\\
        0 & 0  & uB  &  uI_{\gamma-\kappa} & 0 \\
        0 & 0  & C &D & I_\delta
      \end{array}\right).
\]

The following result is  some  equivalent conditions of separable $\mathbb{Z}_2R$ linear self-dual codes.
\begin{theorem}\label{equivalentseparable}
Let $\mathcal{C}$ be a $\mathbb{Z}_2R$  self-dual code of type $(2\kappa,\beta;\beta+\kappa-2\delta,\delta;\kappa)$. Then we have the following statements are equivalent:
\begin{enumerate}
  \item $\mathcal{C}$ is separable.
  \item $\mathcal{C}_X$ is binary self-orthogonal.
  \item $\mathcal{C}_X$ is binary self-dual.
  \item $\mid\mathcal{C}_X\mid=2^\kappa$.
  \item $\mathcal{C}_Y$  is a  self-orthogonal code.
  \item $\mathcal{C}_Y$  is a  self-dual code.
  \item $\mid\mathcal{C}_Y\mid=2^\beta$.
\end{enumerate}
\end{theorem}
\begin{proof} By the similar proof in \citep[Theorem~$3$]{Borgesselfdual}, we are done.\end{proof}

Note that for any vectors $\mathbf{x}$, $\mathbf{y}\in \mathcal{C}_X$, we have
 \[wt_H(\mathbf{x}+\mathbf{y})=wt_H(\mathbf{x})+wt_H(\mathbf{y})-2wt_H(\mathbf{x}*\mathbf{y}),\]where $\mathbf{x}*\mathbf{y}$ is the componentwise product of $\mathbf{x}$ and $\mathbf{y}$. If $wt_H(\mathbf{x})$, $wt_H(\mathbf{y})$, $wt_H(\mathbf{x}+\mathbf{y})$ are all doubly-even, then $wt_H(\mathbf{x}*\mathbf{y})\equiv 0  ~({\rm mod}~2)$, i.e. $\mathbf{x}$ and $\mathbf{y}$ are orthogonal. Therefore,  by Theorem~\ref{equivalentseparable}, we have below result.
\begin{cor}
Let $\mathcal{C}$ be a $\mathbb{Z}_2R$ linear self-dual code. If $\mathcal{C}_X$ has all weights doubly-even, then $\mathcal{C}$ is separable.
\end{cor}
\begin{cor}\label{cor:delta=0}
Let $\mathcal{C}$ be a $\mathbb{Z}_2R$ linear self-dual code of type $(\alpha,\beta;\gamma,\delta;\kappa)$. If $\delta=0$, then $\mathcal{C}$ is separable.
\end{cor}
\begin{proof} If $\delta=0$, then we have for any $(\mathbf{v}|\mathbf{w})\in \mathcal{C}$, $\mathbf{w}$ does not contain unit. Thus, $ \mathcal{C}_Y$ is self-orthogonal, by Theorem~\ref{equivalentseparable}, we get $\mathcal{C}$ is separable.\end{proof}

The above corollary implies the following result.
\begin{cor}
Let $\mathcal{C}$ be a $\mathbb{Z}_2R$ linear self-dual code of type $(\alpha,\beta;\gamma,\delta;\kappa)$. If   $\mathcal{C}$ is non-separable, then $\delta\geq1$.
\end{cor}
\begin{define}
If a $\mathbb{Z}_2R$ linear self-dual code has odd weights, then it is said to be Type $0$. If
it has only even weights, then the code is said to be Type I. If all the codewords have doubly-even weight then it is said to be Type II.
\end{define}

By Lemma~\ref{wandNwareeven}, we get the following  result.
\begin{theorem}
There donot exist $\mathbb{Z}_2R$ linear self-dual codes of   Type $0$.
\end{theorem}
The followings examples are of Type I and Type II.
\begin{example}(Type~I, separable).\label{exa1}
Let $\mathcal{C}$ be a  code of type  $(2,1;2,0;1)$ over $\mathbb{Z}_{2}R$, whose generator matrix is
$$\left(
   \begin{array}{cc|c}
     1 & 1 & 0 \\
     0 & 0 & u \\
   \end{array}
 \right)
$$
It is easy to see $\mathcal{C}=\mathcal{D}\times \mathcal{E}$, where $\mathcal{D}=\langle(11)\rangle$, $\mathcal{E}=\langle(u)\rangle$. Thus $\mathcal{C}$ is separable. Moreover, we have the weight enumerator of this code is
$$W(x,y)=x^{4}+2x^{2}y^{2}+y^{4}.$$
Hence, $\mathcal{C}$ is a self-dual code and  $\mathcal{C}$ is of Type I.
\end{example}
\begin{example}(Type~I, non-separable).
Let $\mathcal{C}$ be generated by the following matrix
$$\left(
   \begin{array}{cccc|ccc}
     1 & 0 & 1 & 0 & 0 & 0 & u \\
     0 & 1 & 0 & 1 & 0 & 0 & u \\
          0 & 0 & 0 & 0 & 0 & u & 0 \\
0 & 0 & 1 & 1 & 1 & 0 & 1+u \\
   \end{array}
 \right)
$$
Clearly, $\mathcal{C}$ is a self-dual code of type $(4,3;3,1;2)$.
Note that $(0~1~0~1)$, $(0~0~1~1)\in \mathcal{C}_{X},$ we have $ \langle(0~1~0~1), (0~0~1~1)\rangle=1$. Thus $\mathcal{C}_{X}$ is not self-orthogonal. By Theorem~\ref{equivalentseparable}, $\mathcal{C} $ is non-separable. Moreover, the weight enumerator of $\mathcal{C}$ is
$$W(x,y)=x^{10}+8x^{8}y^{2}+14x^{4}y^{6}+8x^{2}y^{8}+y^{10}.$$
Therefore, $\mathcal{C} $ is of Type~I code.
\end{example}
\begin{example}(Type~II, separable).
Let $\mathcal{C}$ be generated by
$$\mathcal{G}_{X}=\left(
                   \begin{array}{cccccccc}
                     1 & 0 & 0 & 0 & 0 & 1 & 1 & 1 \\
                     0 & 1 & 0 & 0 & 1 & 0 & 1 & 1 \\
                     0 & 0 & 1 & 0 & 1 & 1 & 0 & 1 \\
                     0 & 0 & 0 & 1 & 1 & 1 & 1 & 0 \\
                   \end{array}
                 \right).
$$
Let $\mathcal{D}$ be generated by
$$\mathcal{G}_{Y}=\left(
   \begin{array}{cccc}
     u & u & 0 & 0 \\
     u & 0 & u & 0 \\
     1 & 1 & 1 & 1 \\
   \end{array}
 \right).
$$
It is easy to check $\mathcal{C}$ is a self-dual code over $\mathbb{Z}_2$  and $\mathcal{D}$ is a self-dual code over $R$. Moreover, the weight of each codeword of $\mathcal{C}$ and $\mathcal{D}$ is doubly-even, respectively. Consider the following matrix
$$\left(
                                                                           \begin{array}{c|c}
                                                                             \mathcal{G}_{X} & \mathbf{0} \\
                                                                             \mathbf{0} & \mathcal{G}_{Y} \\
                                                                           \end{array}
                                                                         \right).
$$ It generates a self-dual code $\mathcal{C}\times\mathcal{D}$ over $\mathbb{Z}_{2}R$.  Since  $\mathcal{C}$ and $\mathcal{D}$ are doubly-even codes, then $\mathcal{C}\times\mathcal{D}$ is of Type II.
\end{example}
\begin{example}\label{ex:2.7}(Type~II, non-separable).
Consider the following matrix
$$\left(
   \begin{array}{cccc|cc}
     1 & 0 & 1 & 0 & 0 & u \\
     0 & 1 & 0 & 1 & 0 & u \\
     0 & 0 & 1 & 1 & 1 & 1+u \\
   \end{array}
 \right),
$$
which generates a linear code $\mathcal{C} $ of type $(4,2;2,1,2)$ over $\mathbb{Z}_2R$. Note that $\mathcal{C} $ is a self-orthogonal code and $|\mathcal{C}|=2^4$, then $\mathcal{C} $ is self-dual. The weight enumerator of this code is
$$W(x,y)=x^{8}+14x^{4}y^{4}+y^{8}.$$
It is easy to see $\mathcal{C}_Y$ is not self-orthogonal. By Theorem~\ref{equivalentseparable}, we have $\mathcal{C} $ is non-separable. To sum up, $\mathcal{C} $ is of Type~II code.
\end{example}

From above examples, we  obtain the minimal  bounds of $\alpha,~ \beta$ for different types.
\begin{theorem}\label{lem:alpha4beta2}
Let $\mathcal{C}$ be a $\mathbb{Z}_2R$ self-dual code of type $(\alpha,\beta;\gamma,\delta;\kappa)$ with $\alpha\cdot\beta>0$.
\begin{itemize}
  \item If $\mathcal{C}$ is Type I and separable, then $\alpha\geq2$, $\beta\geq1$.
  \item If $\mathcal{C}$ is Type I and non-separable, then $\alpha\geq4$, $\beta\geq2$.
  \item If $\mathcal{C}$ is Type II, then $\alpha\geq4$, $\beta\geq2$.
\end{itemize}
\end{theorem}
\begin{proof} If $\mathcal{C}$ is Type I and separable, then $\mathcal{C}_X$ is binary self-dual and $\mathcal{C}_X$ is  self-dual over $R$. Thus $\alpha\geq2$, $\beta\geq1$.

If $\mathcal{C}$ is Type I and non-separable.  By Lemma~\ref{wandNwareeven}, we have $\mathcal{C}_x$ are even weight. If $\alpha=2$, then $\mathcal{C}_x$ is  binary self-dual. By Theorem~\ref{equivalentseparable}, it is a contradiction. Thus $\alpha\geq4$. From Example~\ref{ex:2.7}, we have $\beta\geq2$.

Assume $\mathcal{C}$ is Type II, by Lemma~\ref{wandNwareeven}, we have $(\mathbf{1}^\alpha,\mathbf{0}^\beta),(\mathbf{0}^\alpha,\mathbf{u}^\beta)\in \mathcal{C}$. Since  $\mathcal{C}$ is Type II, then $\alpha\equiv0~({\rm mod}~ 4)$, $\beta\equiv0~({\rm mod}~ 2)$. Note that Example~\ref{ex:2.7}, then $\alpha\geq4$, $\beta\geq2$.
\end{proof}

\section{Several constructions of self-dual codes}
In this section, we present several kinds of construction methods for self-dual
codes over $\mathbb{Z}_2R$.
The following theorem is the first one that  self-dual codes over $\mathbb{Z}_2R$ are obtained from other self-dual codes over $\mathbb{Z}_2R$.
\begin{theorem}\label{th:G1G2}
Let $\mathcal{C}$ be a self-dual code  with generator matrix $G=(G_1| G_2)$ of type $(\alpha,\beta;\gamma,\delta;\kappa)$, and $\mathcal{C}'$ be a self-dual code with generator matrix $G'=(G'_1|G'_2)$ of type $(\alpha',\beta';\gamma',\delta';\kappa')$. Then
\[ \left( \begin{array}{cc|cc}
            G_1&0 &G_2 &0 \\
            0& G'_1 & 0&  G'_2
          \end{array}
\right)\]
generates a  self-dual code $\mathcal{M}$ of type $(\alpha+\alpha',\beta+\beta';\gamma+\gamma',\delta+\delta';\kappa+\kappa')$. Moreover, we have the weight enumerator of $\mathcal{M}$ is
$$W_\mathcal{M}(X,Y)=W_\mathcal{C}(X,Y)W_{\mathcal{C}'}(X,Y).$$
\end{theorem}
\begin{proof} For any codeword $(\mathbf{v}| \mathbf{w})\in \mathcal{M}$, we  have
\[ (\mathbf{v}| \mathbf{w})=\left(A_{1\times (\gamma+\delta)},A'_{1\times (\gamma'+\delta')}\right)\left( \begin{array}{cc|cc}
            G_1&0 &G_2 &0 \\
            0& G'_1 & 0&  G'_2
          \end{array}
\right),\]
where  $\left(A_{1\times (\gamma+\delta)},A'_{1\times (\gamma'+\delta')}\right)\in R^{\gamma+\gamma'+\delta+\delta'}$. Note that any two rows of generator $\left( \begin{array}{cc|cc}
            G_1&0 &G_2 &0 \\
            0& G'_1 & 0&  G'_2
          \end{array}
\right)$  are orthogonal, then $\mathcal{M}$ is a self-orthogonal code. Since $\mathcal{C}$  and $\mathcal{C}'$ are self-dual codes, then $\alpha+2\beta=2(\gamma+2\delta)$ and $\alpha'+2\beta'=2(\gamma'+2\delta')$. Note that the length of $\mathcal{M}$ is $\alpha+\alpha'+2\beta+2\beta'$ and $\mid \mathcal{M}\mid=2^{\gamma+\gamma'+2\delta+2\delta'}$, then we have $\mathcal{M}$ is a self-dual code of type $(\alpha+\alpha',\beta+\beta';\gamma+\gamma',\delta+\delta';\kappa+\kappa')$. It is easy to check that
$$W_\mathcal{M}(X,Y)=W_\mathcal{C}(X,Y)W_{\mathcal{C}'}(X,Y).$$
\end{proof}
\begin{cor}
There exist $\mathbb{Z}_2R$ linear self-dual codes of type $(\alpha,\beta;\gamma,\delta;\kappa)$ for all even $\alpha$ and all $\beta$.
\end{cor}

In the following, we first establish a relationship between    $\mathbb{Z}_2R$ linear self-dual codes and  $\mathbb{Z}_2\mathbb{Z}_4$-additive self-dual codes. From this relationship, then we can construct self-dual codes over $\mathbb{Z}_2R$ form self-dual codes over $\mathbb{Z}_2\mathbb{Z}_4$, or self-dual codes over $\mathbb{Z}_2\mathbb{Z}_4$ form self-dual codes over $\mathbb{Z}_2R$.

Define a map
$$ \theta: \mathbb{Z}_2R \longrightarrow \mathbb{Z}_2\mathbb{Z}_4$$
$$(\mathbf{v}| \mathbf{w}) \longmapsto (\mathbf{v}|\theta(\mathbf{w})),$$
where $(\mathbf{v}|\theta(\mathbf{w}))=(\mathbf{v}| \theta(w_1),\cdots, \theta(w_\beta))$,  $~\theta(0)=0, ~\theta(1)=1, ~\theta(u)=2,~ \theta(1+u)=3.$
\begin{theorem}\label{th:relationZ2Z2utoZ2Z4}
Let $\mathcal{C}$ be a $\mathbb{Z}_2R$  linear code with generator matrix $G$, and for any two codewords $(\mathbf{v}| \mathbf{w}),(\mathbf{x}| \mathbf{y})\in \mathcal{C}$ with  $4 \mid N_{1,1}(\mathbf{w},\mathbf{y})$. If  $\mathcal{C}$ is a self-orthogonal  code, then $\theta(\mathcal{C})$ is also a self-orthogonal code over $\mathbb{Z}_2\mathbb{Z}_4$. Furthermore, if $\mathcal{C}$ is a self-dual  code, then $\theta(\mathcal{C})$ is also a self-dual code over $\mathbb{Z}_2\mathbb{Z}_4$.
\end{theorem}
\begin{proof} Since $\mathcal{C}$ is a self-orthogonal code, then for any two codewords $(\mathbf{v}| \mathbf{w})$ and $(\mathbf{x}| \mathbf{y})\in \mathcal{C}$, we have
\[\langle(\mathbf{v}| \mathbf{w}),(\mathbf{x}| \mathbf{y})\rangle=u\sum_{i=1}^\alpha v_ix_i+\sum_{j=1}^\beta w_jy_j=0\in R. \]
Since $4 \mid N_{1,1}(\mathbf{w})$ and $N_{1,1}(\mathbf{w},\mathbf{y})=N_s(\mathbf{w},\mathbf{y})+N_d(\mathbf{w},\mathbf{y})$, following the notations in Section $2$, we get
\begin{eqnarray*}
    & & \langle(\mathbf{v}| \mathbf{w}),(\mathbf{w},\mathbf{y})\rangle \\
    &=& u\sum\limits_{i=1}^\alpha v_ix_i+uN_{1,u}(\mathbf{w},\mathbf{y})+uN_{u,1}(\mathbf{w},\mathbf{y})+N_s(\mathbf{w},\mathbf{y})+N_d(\mathbf{w},\mathbf{y})+uN_d(\mathbf{w},\mathbf{y}) \\
 &=& u\left(\sum\limits_{i=1}^\alpha v_ix_i+N_{1,u}(\mathbf{w},\mathbf{y})+N_{u,1}(\mathbf{w},\mathbf{y})+
N_d(\mathbf{w},\mathbf{y})\right)+N_s(\mathbf{w},\mathbf{y})+N_d(\mathbf{w},\mathbf{y})\\
&=& u\left(\sum\limits_{i=1}^\alpha v_ix_i+N_{1,u}(\mathbf{w},\mathbf{y})+N_{u,1}(\mathbf{w},\mathbf{y})+
N_d(\mathbf{w},\mathbf{y})\right)=0\in R,
\end{eqnarray*}
which implies that  $2\mid \left(\sum\limits_{i=1}^\alpha v_ix_i+N_{1,u}(\mathbf{w},\mathbf{y})+N_{u,1}(\mathbf{w},\mathbf{y})+
N_d(\mathbf{w},\mathbf{y})\right)$.
Note that
\begin{eqnarray*}
   & &\langle(\mathbf{v}|\theta(\mathbf{w})),(\mathbf{x}| \theta(\mathbf{y}))\rangle \\ &=&2\sum\limits_{i=1}^\alpha v_ix_i+2N_{1,u}(\mathbf{w},\mathbf{y})+2N_{u,1}(\mathbf{w},\mathbf{y})+N_s(\mathbf{w},\mathbf{y})
+N_d(\mathbf{w},\mathbf{y})+2N_d(\mathbf{w},\mathbf{y}) \\
     & =&2\left(\sum\limits_{i=1}^\alpha v_ix_i+N_{1,u}(\mathbf{w},\mathbf{y})+N_{u,1}(\mathbf{w},\mathbf{y})+N_d\right)+N_s(\mathbf{w},\mathbf{y})+N_d(\mathbf{w},\mathbf{y})\\
     &  =&2\left(\sum\limits_{i=1}^\alpha v_ix_i+N_{1,u}(\mathbf{w},\mathbf{y})+N_{u,1}(\mathbf{w},\mathbf{y})+N_d(\mathbf{w},\mathbf{y})\right)+N_{1,1}(\mathbf{w},\mathbf{y}).
  \end{eqnarray*}
Since $4\mid N_{1,1}(\mathbf{w},\mathbf{y})$  and $2\mid \left(\sum\limits_{i=1}^\alpha v_ix_i+N_{1,u}(\mathbf{w},\mathbf{y})+N_{u,1}(\mathbf{w},\mathbf{y})+
N_d(\mathbf{w},\mathbf{y})\right)$, then
\[\langle(\mathbf{v}|\theta(\mathbf{w})),(\mathbf{x}| \theta(\mathbf{y}))\rangle   =0\in \mathbb{Z}_4.\] Therefore, we have  $\theta(\mathcal{C})$ is  a self-orthogonal code over $\mathbb{Z}_2\mathbb{Z}_4$.

Furthermore, if  $\mathcal{C}$ is a self-dual  code with generator matrix $G$, we have $\langle\theta(\mathcal{C}),\theta(G)\rangle=0$. Hence, $\langle\theta(\mathcal{C}),\langle\theta(G)\rangle\rangle=0$, which implies that $\theta(\mathcal{C})\subseteq \langle\theta(G)\rangle^{\perp}$. Note that $\theta$ is a bijective, then $$|\theta(\mathcal{C})|=2^{\gamma+2\delta},~~~|\langle\theta(G)\rangle^{\perp}|=2^{\alpha+
2\beta}/|\langle\theta(G)\rangle|=2^{\gamma+2\delta}.$$ This implies that $\theta(\mathcal{C})$ is a $\mathbb{Z}_2R$ linear code. Together with $\theta(\mathcal{C})$ is self-orthogonal and $|\theta(\mathcal{C})|=2^{\gamma+2\delta}$, we have $\theta(\mathcal{C})$ is also self-dual.\end{proof}

Similar as proof in Theorem~\ref{th:relationZ2Z2utoZ2Z4}, we have below result.
\begin{theorem}\label{th:relationZ2Z4toZ2Z2u}
Let $\mathcal{C}$ be a $\mathbb{Z}_2R$  linear code with generator matrix $G$, and for any codewords $(\mathbf{v}| \mathbf{w}),(\mathbf{x}| \mathbf{y})\in \mathcal{C}$ such that  $4 \mid N_{1,1}(\mathbf{w},\mathbf{y})$. If  $\mathcal{C}$ is a self-orthogonal  code, then $\theta^{-1}(\mathcal{C})$ is also a self-orthogonal code over $\mathbb{Z}_2R$. Furthermore, if $\mathcal{C}$ is a self-dual  code, then $\theta^{-1}(\mathcal{C})$ is also a self-dual code over $\mathbb{Z}_2R$.
\end{theorem}

For convience, we Let  $\mathcal{C}$ be a self-dual code of length $\ell=\alpha+\beta$ and $(G_0|G_1)=(\mathbf{g}_{i}|\mathbf{r}_i)$ be the generator matrix of $\mathcal{C}$, where $\mathbf{g}_{i}$($\mathbf{r}_i$) is the $i$-th row of $G_0$($ G_1$), $1\leq i\leq \gamma+\delta$, respectively.
In the following, we use a building-up approach  to construct  self-dual codes of different lengths.  Here, we just give the proof of Theorem~\ref{th:con1}, since Theorem~\ref{th:con2} and Theorem~\ref{th:con3} can be proved by similar way.
\begin{theorem}\label{th:con1}
Assume the notations given as above. Let $\mathbf{x}$ be a vector in $\mathbb{Z}_2^\alpha$ with odd $wt_H(\mathbf{x})$ and $\mathbf{y}$ be a vector in $u\mathbb{Z}_2^\beta$ with $\langle\mathbf{r_i},\mathbf{y}\rangle=0$ for $1\leq i\leq \gamma+\delta$. Suppose that $h_i=\langle\mathbf{g_i},\mathbf{x}\rangle$ for $1\leq i\leq \gamma+\delta$. Then the following matrix
$$G=\left(\begin{array}{ccc|c}
      1 & 0 & \mathbf{x} & \mathbf{y} \\ \hline
      h_1 & h_1 & \mathbf{g}_1 & \mathbf{r}_1 \\
      \vdots & \vdots &  \vdots & \vdots \\
      h_{\gamma+\delta} &  h_{\gamma+\delta} & \mathbf{g}_{\gamma+\delta} & \mathbf{r}_{\gamma+\delta}
    \end{array}\right)
$$
generates a self-dual code $\mathcal{D}$ over $\mathbb{Z}_2R$ of length $\ell+2$. Moreover, we have $\mathcal{C}$ is separable if and only if $\mathcal{D}$ is separable.
\end{theorem}
\begin{proof}
It is easy to obtain
$$GG^T=0.$$
Hence, $\mathcal{D}$ is   self-orthogonal. Since
$|\mathcal{D}|\cdot |\mathcal{D}^\perp|=2^{\alpha+2+2\beta},  $
so $\mathcal{D}$ is   self-dual if and only if
$$|\mathcal{D}|=2^{\frac{\alpha+2+2\beta}{2}}.$$

Now, we begin computing the size of code $\mathcal{D}$. We first assume that there exist two vectors $(\mathbf{v}_1|\mathbf{w}_1)$ and $(\mathbf{v}_2|\mathbf{w}_2)$ such that
\begin{equation*}\label{eq:zhengmingweiyi}
(\mathbf{v}_1|\mathbf{w}_1)G=(\mathbf{v}_2|\mathbf{w}_2)G, ~\mbox{where}~\mathbf{v}_i\in \mathbb{Z}_2,~\mathbf{w}_i\in R^{\gamma+\delta},~i=1,2,
\end{equation*}
i.e.
$$(\mathbf{v}|\mathbf{w})G=0 , ~~\mbox{where}~\mathbf{v}=\mathbf{v}_1-\mathbf{v}_2,~ \mathbf{w}=\mathbf{w}_1-\mathbf{w}_2 .$$
Then
$$\mathbf{v}(1~0~\mathbf{x}~|~\mathbf{y})+\mathbf{w}\left(\begin{array}{ccc|c}
      h_1 & h_1 & \mathbf{g}_1 & \mathbf{r}_1 \\
      \vdots & \vdots &  \vdots & \vdots \\
      h_{\gamma+\delta} &  h_{\gamma+\delta} & \mathbf{g}_{\gamma+\delta} & \mathbf{r}_{\gamma+\delta}
    \end{array}\right)=0.$$
We have
$$\mathbf{v}+\mathbf{w}\left(\begin{array}{c}
       h_1  \\
      \vdots  \\
     h_{\gamma+\delta}
    \end{array}\right)=0
~\mbox{and}~
 \mathbf{w}\left(\begin{array}{c}
       h_1  \\
      \vdots  \\
     h_{\gamma+\delta}
    \end{array}\right)=0. $$
This implies that $\mathbf{v}=0$. Hence, we have
\begin{equation}\label{eq:liangge}
 \mathbf{w}\left(\begin{array}{c|c}
       g_1  &r_1  \\
      \vdots& \vdots \\
     g_{\gamma+\delta}&r_{\gamma+\delta}
    \end{array}\right)=0.
\end{equation}
Note that
$$\left(\begin{array}{c|c}
     \mathbf{g}_1 & \mathbf{r}_1 \\
        \vdots & \vdots \\ \mathbf{g}_{\gamma+\delta} & \mathbf{r}_{\gamma+\delta}
    \end{array}\right)$$
is the generator matrix of $\mathcal{C}$, then from equation~\eqref{eq:liangge}, we have $\mathbf{w}=0$. This shows that
$$ \mathcal{D}=2^{\gamma+1+2\delta}.$$
Note that $\mathcal{C}$ is self-dual, then $\alpha+2\beta=2(\gamma+2\delta).$ Together with $|\mathcal{D}|=2^{\frac{\alpha+2+2\beta}{2}}$, we obtain $\mathcal{D}$ is self-dual.

It is easy to see
$\mathbf{y}$ is orthogonal to $\mathbf{y}$ and $\mathbf{r}_i$, where $1\leq i\leq \gamma+\delta$. Therefore, by Theorem~\ref{equivalentseparable}, we have $\mathcal{D}$ is separable if and only if $\mathcal{C}$ is separable.
\end{proof}
\begin{theorem}\label{th:con2}
Assume the notations given as above. Let $\mathbf{y}$ be a vector in $R^\beta$ with odd $ wt_L(\mathbf{y})$ and $\mathbf{x}$ be a vector in $ \mathbb{Z}_2^\alpha$ such that $wt_H(x)$ is even and $\langle\mathbf{g_i},\mathbf{x}\rangle=0$ for $1\leq i\leq \gamma+\delta$. Let $t$ be a unit in $R$ and $s_i=\langle\mathbf{r_i},\mathbf{y}\rangle$ for $1\leq i\leq \gamma+\delta$. Then the following matrix
$$G=\left(\begin{array}{c|ccc}
     \mathbf{x}& 1 & 0 &   \mathbf{y} \\ \hline
    \mathbf{g}_1 & s_1 & ts_1 & \mathbf{r}_1   \\
      \vdots & \vdots &  \vdots & \vdots \\
      \mathbf{g}_{\gamma+\delta}& s_{\gamma+\delta} &  ts_{\gamma+\delta} & \mathbf{r}_{\gamma+\delta}
    \end{array}\right)
$$
generates a self-dual code $\mathcal{E}$ over $\mathbb{Z}_2R$ of length $m+2$. Moreover, we have $\mathcal{E}$ is separable if and only if $\mathcal{C}$ is separable.
\end{theorem}

\begin{theorem}\label{th:con3}
Assume the notations  given as above. Let $\mathbf{x}$ be a vector in $\mathbb{Z}_2^\alpha$ with odd $wt_H(\mathbf{x})$ and $\mathbf{y}$ be a vector in $R^\beta$ with odd $wt_L(\mathbf{y})$. Let $\mathbf{e}$ be a vector in $ \mathbb{Z}_2^\alpha$ such that $wt_H(\mathbf{e})$ is even and $\langle\mathbf{g_i},\mathbf{e}\rangle=0$ for $1\leq i\leq \gamma+\delta$, and $\mathbf{a}$ be a vector in $u \mathbb{Z}_2^\beta$ satisfying $\langle\mathbf{r_i},\mathbf{a}\rangle=0$ for $1\leq i\leq \gamma+\delta$. Suppose that $\langle (\mathbf{x}|\mathbf{a}),(\mathbf{e}|\mathbf{y})\rangle=0$, $h_i=\langle\mathbf{g_i},\mathbf{x}\rangle$, $s_i=\langle\mathbf{r_i},\mathbf{y}\rangle$ for $1\leq i\leq \gamma+\delta$. Let $t$ be a unit in $R$, then the following matrix
$$G=\left(\begin{array}{ccc|ccc}
  1 & 0& \mathbf{x}& 0 & 0 &   \mathbf{a} \\
0 & 0&\mathbf{e}& 1 & 0 &   \mathbf{y} \\ \hline
  h_1 & h_1& \mathbf{g}_1  & s_1 & ts_1 & \mathbf{r}_1   \\
  \vdots & \vdots & \vdots & \vdots &  \vdots & \vdots \\
  h_{\gamma+\delta}  &h_{\gamma+\delta} & \mathbf{g}_{\gamma+\delta}& s_{\gamma+\delta} &  ts_{\gamma+\delta} & \mathbf{r}_{\gamma+\delta}
    \end{array}\right)
$$
generates a self-dual code $\mathcal{F}$ over $\mathbb{Z}_2R$ of length $m+4$. Furthermore, if $\langle\mathbf{x},\mathbf{e}\rangle=0\in \mathbb{Z}_2$, we have $\mathcal{F}$ is separable if and only if $\mathcal{C}$ is separable. Otherwise, we have $\mathcal{F}$ is non-separable.
\end{theorem}

%From Theorems~\ref{th:relationZ2Z2utoZ2Z4} and \ref{th:relationZ2Z4toZ2Z2u}, the following result is obtained.
%\begin{cor}
%Assume the notations given as above. Then the number of self-dual codes $\mathcal{C}_u$   over $\mathbb{Z}_2\mathbb{Z}_2[u]$ is same as self-dual codes $\mathcal{C}_4$ over  $\mathbb{Z}_2\mathbb{Z}_4$.
%\end{cor}

%In the following, we give a construction of self-dual codes over $\mathbb{Z}_2R$ from existent self-dual codes.

%Let $A_i^\mathcal{C}$ denote the number of codewords with Lee weight $i$ in $\mathcal{C}$, and $A_i^\mathcal{D}$ denote the number of codewords with Lee weight $i$ in $\mathcal{D}$. Denote $S_\mathcal{C}$ ($S_\mathcal{D}$) is the set of all Lee weight of $\mathcal{C}$ ($\mathcal{D}$), respectively.

%\begin{theorem}\label{th:constructionselfdual}
%Let $G_{m\times \alpha}$ generater a binary self-orthogonal code in $\mathbb{Z}_2^\alpha$, and $\mathcal{G}_{m\times \beta}$ be a matrix with all elements of set  $\{0,u\}$. Then we have that \[\left( \begin{array}{c|c}
%                     G_{m\times \alpha} & \mathcal{G}_{m\times \beta} \\
%                     \mathcal{M}_1 & \mathcal{M }_2
%                   \end{array}
%   \right)\]
%generates a self-dual code over $\mathbb{Z}_2^\alpha\times \mathbb{Z}_2[u]^\beta$, where $\mathcal{M}_1$, $\mathcal{M }_2$ satisfy the following conditions:
%\begin{itemize}
%  \item for any row $x\in \mathcal{M}_1$, $wt_H(x)$ is even,
%  \item ?
%  \item ?
%\end{itemize}
%\end{theorem}

%%%%%%%%%%%%%%%%%%%%%%%%%%%%%%%%%%%%%%%%%%%%%%%%%%%%%%%%%
\section{The structure  of two-weight self-dual codes}
Note that if $\mathcal{C}$ is self-dual, then $(\mathbf{1}^\alpha|\mathbf{0}^\beta)$, $ (\mathbf{0}^\alpha|\mathbf{u}^\beta)\in\mathcal{C}$, where $\alpha \cdot\beta\neq 0$.  This implies that $\mathcal{C}$ has at least two different weights. So, we study a class of self-dual codes with two different weights and get the structure  of these self-dual  codes. Now, we  assume that  $\mathcal{C}$ is a   two-weight  linear code  over $\mathbb{Z}_2R$ with $\alpha \cdot\beta\neq 0$, and $(\mathbf{1}^\alpha|\mathbf{0}^\beta)$, $ (\mathbf{0}^\alpha|\mathbf{u}^\beta)\in\mathcal{C}$.   Since $(\mathbf{1}^\alpha|\mathbf{0}^\beta)$, $ (\mathbf{0}^\alpha|\mathbf{u}^\beta)$ and  $(\mathbf{1}^\alpha|\mathbf{u}^\beta)$ are all in $\mathcal{C}$, then  $\alpha=2\beta$. Note that
$\alpha+2\beta=n$, then $\alpha=n/2,$ $\beta=n/4$. Note that if $4\nmid n$, these codes do not exist. Thus,
 we assume $4\mid n$ in this section.
%\begin{define}
%Let $\mathcal{C}$ be a nonzero code. If all of its nonzero codewords just have two weights, then $\mathcal{C}$ is called two weight code. In this case, if the two weights is $m_{1},~m_{2}$, then the code is called a two weight code with weight $m_{1},~m_{2}$.
%\end{define}
\begin{lemma}\label{lem:weightdis}
Assume the notations given as above. Then the weight distribution of linear code  $\mathcal{C}$ is obtained as follows
\begin{table}[!h]
  \centering
\begin{tabular}{c|c}
  \hline
  Weight & Frenquence\\\hline
$0$ & $1$ \\ \hline
 $n$ & $1$ \\ \hline
$\frac{n}{2}$ & $|\mathcal{C}|-2$ \\
  \hline
\end{tabular}
\end{table}
 \end{lemma}

We define a  linear subcode of $\mathcal{C}$ as follows
\[ \mathcal{C}^*=\{(\mathbf{0}^\alpha|\mathbf{0}^\beta),(\mathbf{1}^\alpha|\mathbf{0}^\beta)
(\mathbf{0}^\alpha|\mathbf{u}^\beta)(\mathbf{1}^\alpha|\mathbf{u}^\beta)\}.\]
\begin{lemma}\label{lem:TW1}
Let $\mathcal{C}$ be a   two-weight  linear code. For any codeword $(\mathbf{v}|\mathbf{w})\in \mathcal{C}\setminus \mathcal{C}^*$, we have
\begin{itemize}
  \item  $N(\mathbf{w})=\frac{n}{4},~N_u(\mathbf{w})=0,~wt_H(\mathbf{v})=\frac{n}{4};$ or
  \item   $N(\mathbf{w})=0,~N_u(\mathbf{w})=\frac{n}{8},~wt_H(\mathbf{v})=\frac{n}{4}.$

\end{itemize}

\end{lemma}
%\setminus\{(\mathbf{1}^\alpha|\mathbf{u}^\beta),(\mathbf{0}^\alpha|\mathbf{0}^\beta)\}
\begin{proof}
If $N(\mathbf{w})\neq0$, then
$$wt_L(u(\mathbf{v}|\mathbf{w}))=wt_L(u \mathbf{w} )=2N(\mathbf{w})=\frac{n}{2}.$$
Thus, $N(\mathbf{w})=\frac{n}{4}$. Note that $wt_L (\mathbf{v}|\mathbf{w}) =\frac{n}{2}$ and $\beta=\frac{n}{4}$, we have $wt_H(\mathbf{v})=\frac{n}{4}$, $N_u(\mathbf{w})=0$.

If $N(\mathbf{w})=0$, then $N_u(\mathbf{w})\neq0$.   Note that $(\mathbf{0}^\alpha|\mathbf{u}^\beta)\in \mathcal{C}$, then
\[wt_L\left( (\mathbf{v}|\mathbf{w})+(\mathbf{0}^\alpha|\mathbf{u}^\beta)\right)=  \frac{n}{2} .\]
Thus $N_u(\mathbf{w})=\frac{n}{8}$, $wt_H(\mathbf{v})=\frac{n}{4}$.
\end{proof}

\begin{example}
Let $\mathcal{C}_1$ be generated by $\left(
                                                        \begin{array}{cc|c}
                                                          1 & 1 & 0 \\
                                                          0 & 0 & u \\
                                                        \end{array}
                                                      \right)
$, then $$\mathcal{C}_1=\{(00|0),(11|0),(00| u),(11| u)\}.$$ It is easy to check that $\mathcal{C}_1$ is a   seif-dual code with two nonzero weights.

Let $\mathcal{C}_2$ be generated by $\left(
                                                        \begin{array}{cc|c}
                                                          1 & 1 & 0 \\
                                                          1 & 0 & 1 \\
                                                        \end{array}
                                                      \right)
$, then $$\mathcal{C}_2=\{(00|0),(11|0),(10| 1),(01| 1),(10|1),(10|1+u),(01|1),(01|1+u)\}.$$ It is a two-weight linear code, which confirms the result in Lemma~\ref{lem:TW1}.
\end{example}

By above discussion, we get the maximal  bound of $\delta$.
\begin{lemma}\label{cor:delta1}
Assume the notations given as above, then  $\delta\leq1$.
\end{lemma}
\begin{proof}
We assume that $(\mathbf{v}|\mathbf{w})\in\mathcal{C} $ with $N(\mathbf{w})\neq0$. By Lemma~\ref{lem:TW1}, we have $N(\mathbf{w})=\frac{n}{4}$. Since $\beta=\frac{n}{4}$,   then we get $\delta\leq1$.
\end{proof}

With above preparation we can obtain the main result in this section.

\begin{theorem}
Let the notations be given as above. Then $\mathcal{C}$ is   self-dual   if and only if the generator matrix of  $\mathcal{C}$ is permutation equivalent to
$$\left(
   \begin{array}{cc|c}
     1 & 1 & 0 \\
     0 & 0 & u \\
   \end{array}
 \right)~{\rm or}~ \left(
   \begin{array}{cccc|cc}
     1 & 1 & 1 & 1 & 0 & 0 \\
     0 & 1 & 0 & 1 & 0 & u \\
     0 & 0 & 1 & 1 & 1 & 1 \\
   \end{array}
 \right).
$$
\end{theorem}
\begin{proof}
Let $G$ be the generator matrix of  $\mathcal{C}$.
By Lemma~\ref{lem:weightdis}, we have the weight  enumerator of $\mathcal{C}$ is
\begin{equation}\label{eq:1}
W_{\mathcal{C}}(X,Y)=X^{n}+(|\mathcal{C}|-2)X^{\frac{n}{2}}Y^{\frac{n}{2}}+Y^{n}.
\end{equation}
By Theorem~\ref{th:MacIden}, the weight enumerator of  $\mathcal{C}^\perp$ satisfies
\begin{eqnarray}\label{eq:2}
\nonumber    & & |\mathcal{C}|W_{\mathcal{C}^{\perp}}(X,Y)= W_{\mathcal{C}}(X+Y,X-Y)  \\
 \nonumber  &=&  (X+Y)^{n}+(|\mathcal{C}|-2)(X^{2}-Y^{2})^{\frac{n}{2}}+(X-Y)^{n} \\
\nonumber&=&  X^{n}+C_{n}^{1}X^{n-1}Y+C_{n}^{2}X^{n-2}Y^{2}+\cdots+C_{n}^{n-1}XY^{n-1}+Y^{n}\\
\nonumber& & \quad +(|\mathcal{C}|-2)\left(X^{n}-C_{\frac{n}{2}}^{1}(X^{2})^{\frac{n}{2}-1}Y^{2}+\cdots+Y^{n}\right) \\
\nonumber& & \qquad\qquad +X^{n}-C_{n}^{1}X^{n-1}Y+C_{n}^{2}X^{n-2}Y^{2}+\cdots-C_{n}^{n-1}XY^{n-1}+Y^{n}\\
\nonumber&= & -(|\mathcal{C}|-2)\left(C_{\frac{n}{2}}^{1}(X^{2})^{\frac{n}{2}-1}Y^{2}-
C_{\frac{n}{2}}^{2}(X^{2})^{\frac{n}{2}-2}Y^{4}+\cdots+C_{\frac{n}{2}}^{\frac{n}{2}-1}X^{2}(Y^{2})^{\frac{n}{2}-1}\right)  \\
\nonumber&& \qquad\qquad \qquad +|\mathcal{C}|(X^{n}+Y^{n})+2C_{n}^{2}X^{n-2}Y^{2}+\cdots+2C_{n}^{n-2}X^{2}Y^{n-2}\\
\nonumber&=& |\mathcal{C}|X^{n}+ \left(2C_{n}^{2}-(|\mathcal{C}|-2)C_{\frac{n}{2}}^{1}\right)X^{n-2}Y^{2}+ \cdots\\
&&  \qquad\qquad\qquad\qquad\ + \left(2C_{n}^{2}-(|\mathcal{C}|-2)C_{\frac{n}{2}}^{\frac{n}{2}-1}\right)X^{2}Y^{n-2}+|\mathcal{C}|Y^{n}.
\end{eqnarray}
If  $\mathcal{C}$ is   self-dual, then $W_{\mathcal{C}^{\perp}}(X,Y)= W_{\mathcal{C}}(X,Y)$.
Comparing the coefficients of both sides of equation~\eqref{eq:2}, we discuss it in two cases.

i)  If  $n=4$, it is easy to check
$$W_{\mathcal{C}^{\perp}}(X,Y)= W_{\mathcal{C}}(X,Y)=X^{4}+2X^{2}Y^{2}+Y^{4}.$$
Note that $\alpha=\frac{n}{2}$, $\beta=\frac{n}{4}$, we have $\alpha=2,$ $\beta=1$. Since $|\mathcal{C}|=2^{\frac{n}{2}}=4$, and $(\mathbf{1}^\alpha|\mathbf{0}^\beta),(\mathbf{0}^\alpha|\mathbf{u}^\beta)\in\mathcal{C} $, thus $G$ is permutation equivalent to $$\mathcal{C}=\langle\left(
   \begin{array}{cc|c}
     1 & 1 & 0 \\
     0 & 0 & u \\
   \end{array}
 \right)\rangle.$$

ii)  If  $n>4$, then the right side of equation~\eqref{eq:2} has at least five terms. The second term of right side of equation~\eqref{eq:2} $\left(2C_{n}^{2}-(|\mathcal{C}|-2)C_{\frac{n}{2}}^{1}\right)X^{n-2}Y^{2}$ is equal to zero. Hence, $n=8$. Simplifying  equation~\eqref{eq:2}, one has
$$W_{\mathcal{C}^{\perp}}(X,Y)=X^{8}+14X^{4}Y^{4}+Y^{8}. $$ Since
$$W_{\mathcal{C}}(X,Y)=X^{8}+(|\mathcal{C}|-2)X^{4}Y^{4}+Y^{8}, $$ and $|\mathcal{C}|=2^{\frac{n}{2}}=16$, then $$W_{\mathcal{C}^{\perp}}(X,Y)=W_{\mathcal{C}}(X,Y).$$ This implies when $n=8$, there exists a self-dual code over $\mathbb{Z}_2R$.
Note that $\mathcal{C}$ is Type II, by Theorem~\ref{lem:alpha4beta2}, we get $\alpha=4$, $\beta=2$. Thus, $\mathcal{C}$ is non-separable. Otherwise, $\mathcal{C}$ is contradiction with two non-zero weights. By  Lemma~\ref{cor:delta1}, we have $\delta\leq1$. If $\delta=0$, by Corollary~\ref{cor:delta=0}, it is a contradiction. Hence,  $\delta=1$, $\gamma=2$.
Let $$G=\left(\begin{array}{c|c}
          \mathbf{v}_1 & \mathbf{w}_1 \\
          \mathbf{v}_2 & \mathbf{w}_2   \\
          \mathbf{v}_3 & \mathbf{w}_3
        \end{array}\right)
$$, where $\mathbf{v}_i\in \mathbb{Z}_2^4$ for $1\leq i\leq3$, $\mathbf{w}_1,\mathbf{w}_2\in\{0,u\}^2$, $\mathbf{w}_3\in R^2$.
By Lemma~\ref{lem:TW1}, we have $N(\mathbf{w}_3)=2,~wt_H(\mathbf{v}_3)=2,$ i.e. $\mathbf{w}_3=(11)$.  Now we fix the sequence of $(\mathbf{v}_3 | \mathbf{w}_3)=(0011|11)$. Note that $(\mathbf{1}^4|\mathbf{0}^2)\in \mathcal{C}$, then we let $(\mathbf{v}_1 | \mathbf{w}_1)=(\mathbf{1}^4|\mathbf{0}^2)$. Finally, we need to choose  $ (\mathbf{v}_2 | \mathbf{w}_2 )$. Since $ (\mathbf{v}_2 | \mathbf{w}_2 )\notin \mathcal{C}^*$, then by Lemma~\ref{lem:TW1}, we have $N_u(\mathbf{w}_2)=1,~wt_H(\mathbf{v}_2)=2$. Since  $\mathcal{C}$ is non-separable, then $wt_H(\mathbf{v}_2*\mathbf{v}_3)$ is odd. To sum up, we get $G$ is permutation equivalent to
$$\left(
   \begin{array}{cccc|cc}
     1 & 1 & 1 & 1 & 0 & 0 \\
     0 & 1 & 0 & 1 & 0 & u \\
     0 & 0 & 1 & 1 & 1 & 1 \\
   \end{array}
 \right)$$

Conversely, it is easy to check.
We finish the proof.
\end{proof}

From equation~\eqref{eq:2}, we get the bound of the minimal distance of dual code $\mathcal{C}^\perp$.
\begin{theorem}
Let the notations be given as above. If $n=\frac{|\mathcal{C}|}{2}$, then the minimal Hamming
distance of $\mathcal{C}^\perp$ is $4$. Otherwise, the minimal Hamming
distance of $\mathcal{C}^\perp$ is $2$.
%In particular, if $n$ contains  odd factors,   the minimal  distance of $\mathcal{C}^\perp$ is $2$.
\end{theorem}
%%%%%%%%%%%%%%%%%%%%%%%%%%%%%%%%%%%%%%%%%%%%%%%%%%%%%%%%%%%

%%%%%%%%%%%%%%%%%%%%%%%%%%%%%%%%%%%%%%%%%%%%%%%%%%%%%%%%%

\end{document}